%% file: Template_arXiv.tex
\documentclass{article}
\usepackage{spconf,amsmath,graphicx,hyperref}
\usepackage{amsfonts}

\usepackage[caption=false]{subfig}

\title{\MakeUppercase{Power-Dominance in Estimation Theory: A Third Pathological Axis}}
\name{S.~S.~Krishna~Chaitanya~Bulusu$^{\dagger}$,~and~Mikko~J.~Sillanpää$^\ddagger$
}
\address{$\dagger$~Centre for Wireless Communications (CWC),\quad$\ddagger$~Department of Mathematical Sciences (DMS),\\University~of~Oulu,~Oulu,~FI-90014,~Finland.}
\usepackage[printonlyused]{acronym}
\input{acrons.tex}

\usepackage{siunitx}

\usepackage{amsthm,amssymb}
\usepackage{mathtools}
\newtheorem{theorem}{Theorem}[section]
\newtheorem{corollary}{Corollary}[theorem]

\newtheorem{definition}{Definition}

\usepackage{booktabs} 

\begin{document}
%
\maketitle
\begin{center}
 \small~\copyright~2025 IEEE. Personal use of this material is permitted. Permission from IEEE must be obtained for all other uses, in any current or future media, including reprinting/republishing this material for advertising or promotional purposes, creating new collective works, for resale or redistribution to servers or lists, or reuse of any copyrighted component of this work in other works.
\end{center}
\begin{abstract}
This paper introduces a novel framework for estimation theory by introducing a second-order diagnostic for estimator design. While classical analysis focuses on the bias-variance trade-off, we present a more foundational constraint. This result is model-agnostic, domain-agnostic, and is valid for both parametric and non-parametric problems, Bayesian and frequentist frameworks. We propose to classify the estimators into three primary power regimes. We theoretically establish that any estimator operating in the `power-dominant regime' incurs an unavoidable mean-squared error penalty, making it structurally prone to sub-optimal performance. We propose a `safe-zone law' and make this diagnostic intuitive through two safe-zone maps. One map is a geometric visualization analogous to a receiver operating characteristic curve for estimators, and the other map shows that the safe-zone corresponds to a bounded optimization problem, while the forbidden `power-dominant zone' represents an unbounded optimization landscape. This framework reframes estimator design as a path optimization problem, providing new theoretical underpinnings for regularization and inspiring novel design philosophies.
\end{abstract}
\begin{keywords}
Bias-variance-power triad, energy geometry, estimation theory, mean-squared error (MSE), optimization, orthogonality, power dominance, regularization, safe-zone map, second-order diagnostic, signal processing.
\end{keywords}
\vspace{-0.5em}
\section{Introduction}
\label{sec:intro}
The \ac{MSE} serves as the primary metric in estimation theory for evaluating the performance of an estimator of a true signal \cite{K93, VT68}. The bias-variance decomposition assesses the fundamental trade-off between systematic error and statistical fluctuation. This trade-off provides a profound and widely adopted lens for analyzing and optimizing estimators across many domains \cite{Hastie09}. However, it may not always reveal structural flaws in an estimator's design \cite{D20}.

This paper presents a new diagnostic and design framework that introduces a third axis of analysis, focusing on the estimator's mean power relative to the true signal. Notwithstanding the bias-variance trade-off, we show that this simple metric reveals a fundamental constraint on estimator performance. Our work exposes a previously unrecognized structural pathology in estimators. In classical estimation theory, an estimator is usually interpreted geometrically within the geometry of an inner product space (refer \cite{K93}, ch. 4, pp. 384–389). We prove that any estimator operating in the `\textit{power-dominant}' regime incurs an unavoidable MSE penalty, making it fundamentally prone to sub-optimal performance. Then, we propose a more practical, diagnostic-based approach to estimator design. 

This paper is organized as follows: Section~\ref{sec:power regime classification and penalty law} introduces the power regimes and proves the power-dominance penalty theorem. Section~\ref{sec:static diagnostic to dynamic design philosophy} presents a dynamic design philosophy based on scaling control. Section~\ref{sec:safe zone maps} presents two safe-zone maps. Section~\ref{sec:relation to prior work} link our work to the prior literature in and conclusions are provided in Section~\ref{sec:conclusion}.
\vspace{-0.5em}
\section{Power Regime Classification System and Power-Dominance Penalty}\label{sec:power regime classification and penalty law}
We introduce a penalty theorem based on a novel classification based on power regimes in this section. Based on this, we establish a fundamental limitation in estimation theory through a geometric interpretation via inner products. This framework reveals a previously unrecognized constraint on estimation performance and provides a new diagnostic for estimator design. 
Previous work has largely focused on bias and variance \cite{K93,VT68}. Our result exposes a hidden trade-off in estimation theory directly linked to power dominance, adding a third axis to classical performance analysis.
\vspace{-0.2em}
\subsection{Power Regimes}\label{sec:2.1. power regime classification}
Let $X$ be a real-valued entity, and let $\hat{X}$ be its estimator.\footnote{Extension to complex-valued entities is straightforward.} The estimation error $e$ and \ac{MSE} are defined as\footnote{In \eqref{eq:mse_def}, the expectation in MSE definition is treated in a measure-theoretic sense, agnostic to frequentist or Bayesian formulations. Our framework unifies both views, focusing instead on the inner-product geometry of estimation.}
\begin{align}
    e &= \hat{X} - X, \label{eq:estimation_error_def} \\
    \mathrm{MSE} &= \mathbb{E}[e^2] = \langle e, e \rangle. \label{eq:mse_def}
\end{align}
While classical minimum \ac{MSE} (MMSE) estimation enforces orthogonality, $\mathbb{E}[\hat{X} \cdot e] = 0$, this property often fails in practice. We examine such departures through the lens of \textit{power dominance} by defining the power regimes.

\begin{definition}[Power Regimes]\label{def:power_regimes}
Let $\hat{X},X \in L^{2}(\Omega;\mathbb{R})$.
\begin{enumerate}
\item \textbf{Power-dominant (Forbidden zone):} $\mathbb{E}[\hat{X}^{2}] > \mathbb{E}[X^{2}]$.
\item \textbf{Power-conservative (Safe zone):} $\mathbb{E}[X^{2}] \geq \mathbb{E}[\hat{X}^{2}]$.
\item \textbf{Power-balance (Balance zone)\footnote{This \textit{power-balance regime} is an edge case of \textit{power-conservative regime}.}:} $\mathbb{E}[\hat{X}^{2}] = \mathbb{E}[X^{2}]$.
\end{enumerate}
\end{definition}
\subsection{Power-dominance Penalty}\label{sec:2.2. power domiance penalty law}
\begin{theorem}[Power-dominance Penalty Theorem]\label{thm:PD_Penalty_theorem}
Let $X,\hat{X} \in L^{2}(\Omega)$ and $e \triangleq \hat{X} - X$. If the estimator operates in the power-dominance regime and $\mathbb{E}[\hat{X}\cdot e] \neq 0$, then
\[
\boxed{\mathbb{E}[\hat{X}\cdot e] > \tfrac12\,\mathrm{MSE}}.
\]
\end{theorem}
\begin{proof}
Using the identity valid in any inner product space \cite{A15},
\begin{align}
&\langle A, A-B\rangle = \nonumber\\
&=\tfrac12\big(\|A\|^2 + \|B\|^2 - 2\langle A, B\rangle + \|A\|^2-\|B\|^2 \big),\nonumber\\
&=\tfrac12\big(\|A\|^2 - \|B\|^2 + \|A-B\|^2\big),
\end{align}
with $A=\hat X$, and $B=X$ as elements in the Hilbert space $L^2(\Omega)$, where $\langle ., .\rangle$ corresponds to expectation \cite{B2005}, we get,
\begin{align}
\langle \hat X, e\rangle
&= \tfrac12\,\mathrm{MSE}+ \tfrac12\big(\mathbb E[\hat X^2] - \mathbb E[X^2]\big). \label{eq:pdpenaltydecomp}
\end{align}

In Regime 1, with $\mathbb{E}[\hat{X}^{2}] > \mathbb{E}[X^{2}]$, from \eqref{eq:pdpenaltydecomp} and \cite{Bulusu2025_arXiv}, under non-degeneracy, we get, $\mathbb E[\hat X\cdot e] >\tfrac12\,\mathrm{MSE}$. The degeneracy condition $\mathbb{E}[\hat{X}\cdot e] \neq 0$ excludes three cases: (i) trivial $\hat{X} = 0$, (ii) optimal $\mathbb{E}[\hat{X}\cdot e] = 0$, and (iii) ideal $\hat{X} = X$.
\end{proof}
Therefore, the \textbf{Theorem~\ref{thm:PD_Penalty_theorem}} indicates that every power-dominant estimator that is operating within an inner space geometry is fundamentally flawed by design, as it incurs an irreducible penalty of at least half the total MSE due to estimator–error entanglement. Moreover, this represents a fundamental limitation in estimation theory, holding independently of assumptions on bias, noise, or model mismatch.

\begin{corollary}[Power-conservative Bound]\label{cor:Power_Conservative_Regime_Bound}
Under the same setup as Theorem~\ref{thm:PD_Penalty_theorem}, if $\mathbb{E}[\hat{X}^2] \leq \mathbb{E}[X^2]$ (Power-conservative regime), then
\[
\boxed{\mathbb{E}[\hat{X}\cdot e] \leq \tfrac12\,\mathrm{MSE}}.
\]
\end{corollary}
\begin{proof}
From \eqref{eq:pdpenaltydecomp}, if $\mathbb{E}[\hat{X}^2] \leq \mathbb{E}[X^2]$, the second term is nonpositive and hence $\mathbb{E}[\hat{X}e] \leq \tfrac12\,\mathrm{MSE}$.
\end{proof}
\noindent
The \textbf{Theorem~\ref{thm:PD_Penalty_theorem}} and \textbf{Corollary~\ref{cor:Power_Conservative_Regime_Bound}} establish that the \textit{mean power of the estimate relative to the true value} is a \textbf{critical diagnostic} for the structural validity of an estimator. This shifts the focus from the classical bias-variance trade-off towards the `bias-variance-power triad' by adding a deeper criterion: \textit{The mean power of the estimate must not exceed that of the signal if orthogonality and low-error performance are to coexist.}
\section{From Static Diagnostic to Dynamic Design Philosophy}\label{sec:static diagnostic to dynamic design philosophy}
Based on \textbf{Definition~\ref{def:power_regimes}}, we introduce a `safe-zone law' through \textbf{Corollary~\ref{cor:Scaling_and_safe_zone_law}} in Section~\ref{sec:3.1. Safe Zone Law}. This law reframes \textit{power-dominance} from a static post-analysis check into a `\textit{dynamic design principle},' which we present in Section~\ref{sec:3.2. Safe-Zone Dynamics}.
\vspace{-0.5em}
\subsection{Safe-zone Law}\label{sec:3.1. Safe Zone Law}
\begin{corollary}[Scaling Reinterpretation and Safe-zone Law]\label{cor:Scaling_and_safe_zone_law}
Under the same setup as Theorem~\ref{thm:PD_Penalty_theorem}, let $Z\in L^{2}$ be any candidate and define $\hat X(t)=tZ$ with error $e(t)=\hat X(t)-X$.
Then the mean-squared error
\begin{equation}
\begin{aligned}
\mathrm{MSE}(t)=\mathbb{E}\!\left[(tZ-X)^{2}\right]
= t^{2}\mathbb{E}[Z^{2}] - 2t\,\mathbb{E}[XZ] + \mathbb{E}[X^{2}]
\label{eq:mse_quadratic}
\end{aligned}
\end{equation}
is minimized at
\begin{equation}
\begin{aligned}
t^{\star}=\frac{\mathbb{E}[XZ]}{\mathbb{E}[Z^{2}]}. \label{eq:tstar}
\end{aligned}
\end{equation}
At $t^{\star}$ the following hold:
\begin{align}
&\text{(Orthogonality)} && \mathbb{E}\!\left[\hat X(t^{\star})\,e(t^{\star})\right]=0, \label{eq:ortho}\\[2pt]
&\text{(Power conservation)} && 
\mathbb{E}\!\left[\hat X(t^{\star})^{2}\right]
=\frac{\mathbb{E}[XZ]^{2}}{\mathbb{E}[Z^{2}]} \ \leq\ \mathbb{E}[X^{2}], \label{eq:power_conserve}
\end{align}
with equality in \eqref{eq:power_conserve} if and only if $Z$ is collinear with $X$.
Consequently, the {MMSE} operating point always lies in the \textit{safe} or \textit{balance regime} $\mathbb{E}[\hat X^{2}]\leq \mathbb{E}[X^{2}]$.
\end{corollary}

\begin{proof}[Abridged proof]
The quadratic form \eqref{eq:mse_quadratic} is minimized at \eqref{eq:tstar}.
Evaluating $\mathbb{E}[\hat X(t)\,e(t)] = t^{2}\mathbb{E}[Z^{2}] - t\,\mathbb{E}[XZ]$ at $t=t^{\star}$ gives \eqref{eq:ortho}.
For \eqref{eq:power_conserve}, compute $\mathbb{E}[\hat X(t^{\star})^{2}] = \mathbb{E}[XZ]^{2}/\mathbb{E}[Z^{2}]$ and apply Cauchy-Schwarz:
$\mathbb{E}[XZ]^{2} \leq \mathbb{E}[X^{2}]\,\mathbb{E}[Z^{2}]$, with equality if and only if $Z$ is collinear with $X$.
This yields $\mathbb{E}[\hat X(t^{\star})^{2}] \leq \mathbb{E}[X^{2}]$.
\end{proof}
\vspace{-0.5em}
\subsection{From Safe-Zone Dynamics to a Scaling–Control Design Philosophy}\label{sec:3.2. Safe-Zone Dynamics}
As per `safe-zone law' (\textbf{Corollary~\ref{cor:Scaling_and_safe_zone_law}}), the progression of the scaling factor \(t\) in an estimator provides a clear, dynamic picture of \ac{MSE} optimization. Thus, the `power–balance line' (i.e. \textit{balance regime}) emerges as a critical operational boundary. Pushing the estimator’s power beyond it is counterproductive to achieving MMSE.\footnote{This conclusion holds for any general scaling function \(f(t)\) provided \(f(t)\) is well-behaved and has a non-zero first derivative.}

To the best of our knowledge, prior work has not identified power dominance as a standalone diagnostic principle or linked it explicitly to a universal operating boundary. The `safe–zone law' formalizes this boundary, showing that the {MMSE} optimum always lies in the \textit{power–conservative regime}, while the power dominance penalty law quantifies the inevitable performance loss when operating in the \textit{power–dominant regime}. Together, these laws establish that estimation performance is governed not only by statistical efficiency but also by \textit{structural energy alignment}. This connection moves the framework from theoretical characterization into a \textit{practical guiding philosophy} that is model-agnostic (parametric or non-parametric, Bayesian or frequentist, data-driven or learning-based) and domain-agnostic (signal processing, statistics, \ac{ML}, control theory, and beyond).

Therefore, estimator design is no longer about passively ``finding the right bias–variance balance,'' but about `{actively controlling the path of scaling}' so that the estimate converges into the safe/balance zone as quickly and accurately as possible. This reframes the problem into a \textbf{scaling–control and path optimization task} with three key well-known objectives:
\begin{enumerate}
    \item \textbf{Speed of Convergence}: Reaching $t^\star$ rapidly, without unnecessary delay, akin to acceleration methods in optimization such as Nesterov schemes \cite{N83}.
    \item \textbf{Overshoot Avoidance}: Preventing overshoot into the forbidden zone, introducing a convergence–stability trade–off analogous to gain tuning in control systems.
    \item \textbf{Dynamic Adaptation}: Tracking a moving $t^\star$ in adaptive or time–varying scenarios, ensuring the estimator remains in the safe-zone despite environmental changes.
\end{enumerate}

This unified perspective yields a simple yet powerful design rule:
\begin{quote}
    \textit{Design an estimator whose scaling converges to the safe–zone optimum with maximum speed and minimum overshoot.}
\end{quote}

Therefore, by reframing an estimation problem as a scaling control and path optimization task, we use language that is easy for experts in optimization and control to understand. This can encourage, in the future, better and more effective ways to investigate long-standing estimation problems. 
\section{Safe-zone Maps for Estimation Theory}\label{sec:safe zone maps}
This Section introduces a new geometric framework for visualizing and diagnosing the structural integrity of estimators. By mapping key performance metrics to a 2-D plane, we provide a universal diagnostic tool analogous to the \ac{ROC} curve in detection theory \cite{F05} and a powerful optimization-based interpretation. This visual representation in Fig~\ref{fig:safe_zone_plots} provides a universal checkpoint for estimator design, immediately flagging any estimator that falls into the structurally compromised \textit{forbidden zone}. In both plots, the \textit{safe-zone} and the \textit{forbidden zone} are shown in green and red colors. 
\vspace{-1em}
\subsection{ROC-like Plot for Estimation Theory}\label{sec:4.1. ROC-like plot}
In stark contrast to detection theory, estimation theory lacks an analogue to the ROC curve. Fig.~\ref{fig:safe_zone_plot1} (i.e., Left map) provides a visual analogy to the ROC curve in detection theory, with distinct performance zones. It provides a high-level diagnostic and introduces a new geometric map based on a second-order diagnostic: the mean power of the estimate relative to the true signal.

Similar to the ROC curve's use for any binary classifier, the plot in Fig.~\ref{fig:safe_zone_plot1} provides a simple and universally applicable visual diagnostic for any estimator. The plot helps to immediately classify an estimator's fundamental behavior without needing a detailed bias-variance decomposition.

In Fig.~\ref{fig:safe_zone_plot1}, the plot's clear ``forbidden" zone provides a direct visual representation of a fundamental performance limitation. Any estimator that falls into this region is subject to an inescapable penalty on the estimator-error coupling term, making it structurally flawed. This is analogous to the ``below-diagonal" region in an ROC curve, which signals that a classifier is worse than random chance.

The plot in Fig.~\ref{fig:safe_zone_plot1} does not just evaluate; it guides. It tells a designer to ``stay in the safe zone." Thereby, the designer shifts the focus from a purely algebraic exercise to a practical, geometric control problem. The `power-balance line,' in particular, is the critical boundary to approach but not cross. Unlike the diagonal in a typical \ac{ROC} curve, which represents a random classifier, this line signifies a state of perfect power balance and is the theoretical limit of optimal performance.
\begin{figure}[t!]
  \centering
  \subfloat[Left map\label{fig:safe_zone_plot1}]
    {\includegraphics[width=0.45\columnwidth]{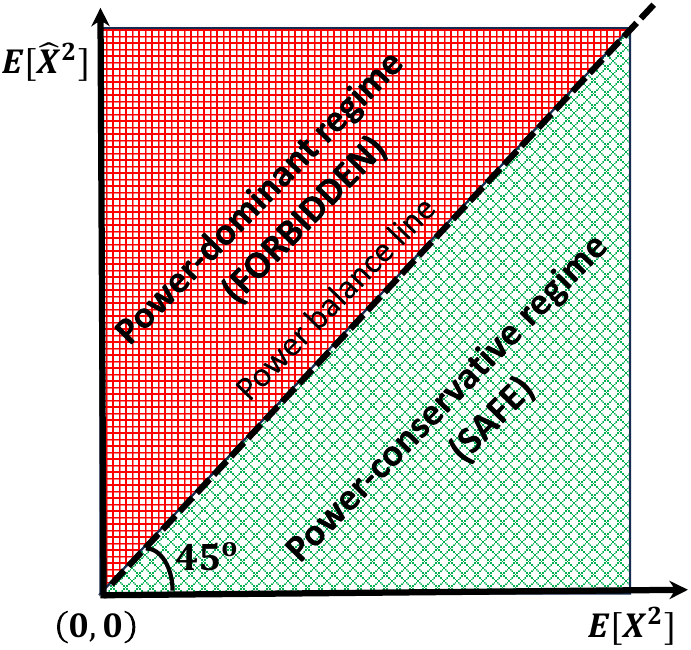}}\quad
  \subfloat[Right map\label{fig:safe_zone_plot2}]
    {\includegraphics[width=0.5\columnwidth]{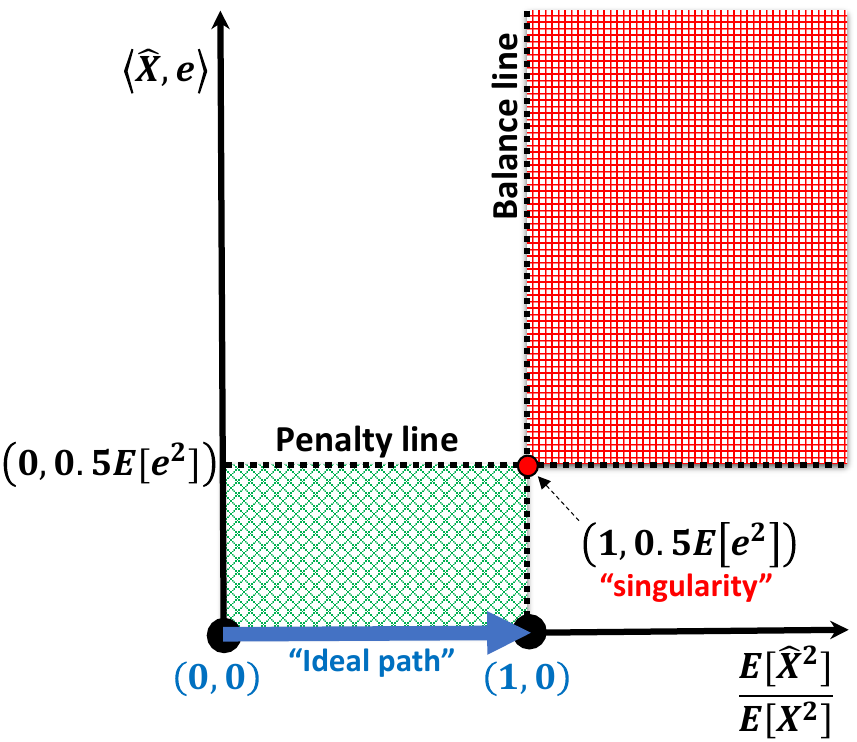}}
  \caption{Safe-zone Maps for Estimation Theory. (a) \textbf{Left~map}:~\ac{ROC}-like visualization of power regimes. (b) \textbf{Right~map}:~Bounded vs.\ unbounded optimization geometry.}
  \label{fig:safe_zone_plots}
  \vspace{-0.5em}
\end{figure}
\vspace{-1em}
\subsection{Optimization Perspective: Bounded vs. Unbounded Regimes}\label{sec:4.1. Optimization theory perspective}
The geometric framework can be further understood through the lens of optimization theory \cite{BT97,BSS06}, revealing why the \textit{power-dominant regime} is fundamentally untenable. We visualize the optimization landscape by plotting the coupling term $\mathbb{E}[\hat{X}\cdot e]$ against the power ratio $\mathbb{E}[\hat{X}^2]/\mathbb{E}[X^2]$ as shown in Fig.~\ref{fig:safe_zone_plot2} (i.e., Right map). The \textit{safe} and \textit{forbidden regimes} are shown in green and red colors, respectively. The `penalty line' and `power-balance line' act as boundaries. Their intersection is the `singularity' point $(0,0.5\text{MSE})$, shown as a red dot in Fig.~\ref{fig:safe_zone_plot2}. Crossing this singularity implies guaranteed energy wastage, futile pursuit, and structural sub-optimality. This reveals a profound distinction between the two regimes:
\begin{itemize}
    \item \textbf{The Bounded, Safe Regime}: The power-conservative region forms a well-posed, bounded optimization problem \cite{S25}. Within this region, the coupling term is constrained by $0 \leq \mathbb{E}[\hat{X}\cdot e] \leq 0.5\text{MSE}$. This well-defined solution space ensures that any optimization effort is meaningful, as the solution is guaranteed to be contained within this feasible region.

    \item\textbf{The Unbounded, Forbidden Regime}: Once the estimator crosses the power-balance line into the power-dominant zone, the optimization problem becomes ill-posed and unbounded. The coupling term suffers from a lower bound of $0.5\text{MSE}$ (the penalty floor), and the energy wastage grows without bound. This signifies a futile optimization exercise where computational or design effort is drained in a search for a solution that does not exist within the feasible estimation space \cite{GJ79, BTN01}.
\end{itemize}

Surprisingly, our framework directly parallels control theory. `\textit{Just as excessive control effort introduces harmful energy into a physical system, and causes instability \cite{Ogata10, FPEN10}, an estimator whose power dominates the true signal injects destructive energy into the error term, causing an inescapable performance penalty and structural pathology}.' This analogy provides a deeper physical intuition for the framework. It shows that a power-dominant estimator is pouring its ``excess energy" not into aligning with the signal. Instead, it is feeding the error term, creating the penalty. 

While the performance and convergence of classical optimization algorithms, such as Newton's method \cite{Nocedal2006}, are highly sensitive to a random initial guess, our formulation is fundamentally different. In Fig. \ref{fig:safe_zone_plot2}, the ideal path is clearly defined as a ray from $(0,0)$ to $(1,1)$ as shown in blue color along the $x$-axis. This journey begins from `maximal epistemic humility,' the `state of ignorance' $(0,0)$, where the estimate is zero and the baseline MSE is E[X²]. It moves toward `optimal accuracy,' the `state of knowledge' $(1,0)$, represented by the MMSE estimator. As shown in Fig. \ref{fig:safe_zone_plot2}, this specific starting point elevates the process from simple computation to a controlled, logical path toward knowledge (1,1), ensuring that integrity is maintained throughout.

\vspace{-1em}
\section{Relation to Prior work}\label{sec:relation to prior work}
While our framework touches upon established concepts in estimation theory, it provides a novel synthesis that distinguishes it from prior work. The classical bias-variance decomposition is a general lens for analyzing estimator performance. Our work does not replace this framework; rather, for a common subclass of estimators adjustable by a multiplicative scale, it supersedes it by providing a more direct, single-step optimization that implicitly handles both bias and variance. The concept of `power-dominance' as a formal diagnostic checkpoint is not explicitly stated in standard literature, which has primarily focused on optimality under bias/variance trade-offs, Bayesian vs. frequentist estimation, and filtering.

Our work provides a theoretical justification for practical methods such as L$^2$ regularization \cite{TA77}, which implicitly controls an estimator's power to prevent overfitting, but without the explicit theoretical link to the `power-dominance' penalty that our framework provides. Similarly, while MMSE estimators inherently operate in the \textit{power-conserving regime}, our work provides a new reason why this is so: their orthogonality principle is a direct consequence of avoiding the `power-dominance' penalty. Thus, our framework unifies these disparate ideas under a new, overarching principle of `energy geometry' in estimation theory. Our framework theoretically justifies, apparently, why standard ML practices succeed: normalization inherently constrains learning to the `safe zone.'
\vspace{-0.5em}
\section{Conclusion}\label{sec:conclusion}
This paper introduces a new framework for estimator design based on a second-order diagnostic: the mean power of the estimate relative to the true signal. We classify estimators into three power regimes and show that power-dominant estimators are structurally prone to an unavoidable MSE penalty. This finding offers a new perspective on estimator design, as traditional bias and variance metrics may not fully reveal an estimator's health. This result extends the classical bias–variance trade-off into a ‘bias–variance–power triad,’ positioning mean power as a structural diagnostic for estimator validity. The proposed framework is visually represented by two safe-zone maps. The first provides a diagnostic similar to an ROC curve, while the second shows that the safe-zone'' corresponds to a bounded, well-posed optimization problem, in stark contrast to the unbounded landscape of the forbidden zone.'' Our work thus reframes estimator design as a path optimization problem in resource allocation and dynamic power control, providing new avenues for research in regularization and adaptive estimation.
\vspace{-1em}
\section{Author Contributions}\label{sec:section_authorcontributions} 
Bulusu conceived the idea, developed the theory, and prepared the manuscript. Sillanpää verified all the proofs, provided crucial insights, and critically reviewed the manuscript.
\vspace{-2em}
\section{Acknowledgement}\label{sec:section_acknowledgement}
The authors acknowledge the use of Google's Gemini for improving the clarity and structure of the manuscript text. All theoretical derivations, analyses, and results were developed and verified manually by the authors.
\bibliographystyle{IEEEbib}
\bibliography{Template_arXiv}

\end{document}

%% file: acrons.tex
\acrodef{3D}{three-dimensional}
\acrodef{3GPP}{3rd generation partnership project}
\acrodef{4G}{fourth generation}
\acrodef{5G}{fifth generation}

\acrodef{B5G}{5G and beyond}
\acrodef{6G}{sixth generation}

\acrodef{DPD}{digital predistortion}
\acrodef{FC}{Fast convergence}
\acrodef{IB}{in-band}
\acrodef{ICE}{iterative compensation of error}
\acrodef{MER}{modulation error ratio}
\acrodef{NL}{nonlinear}
\acrodef{OFDM}{orthogonal frequency division multiplexing}
\acrodef{OOB}{out-of-band}
\acrodef{PA} {power amplifier}
\acrodef{PSD}{power spectral density}
\acrodef{PW-DPD}{piece-wise DPD}

\acrodef{RF} {radio frequency}
\acrodef{DVB-T2} {second generation Digital Video Broadcast-Territorial}
\acrodef{ATSC 3.0} {American next generation DTT}
\acrodef{ICE} {iterative compensation of error}
\acrodef{ILA} {indirect learning architecture}
\acrodef{FCJO} {FC joint optimization}
\acrodef{PAPR} {peak-to-average power ratio}
\acrodef{PW} {piecewise}
\acrodef{IFFT} {inverse fast Fourier transform}
\acrodef{FFT} {fast Fourier transform}
\acrodef{MP} {memory polynomial}
\acrodef{GMP} {generalized memory polynomial}
\acrodef{LUT} {look-up table}

\acrodef{ML}{machine learning}
\acrodef{OBO}{output backoff}
\acrodef{ARB}{arbitrary waveform generator}

\acrodef{NR}{new radio}
\acrodef{QAM}{quadrature amplitude modulation}
\acrodef{FLOPs}{floating point operations}
\acrodef{DSP}{digital signal processing}
\acrodef{IBO}{input back-off}
\acrodef{OBO}{output back-off}
\acrodef{VSA}{Vector Signal Analyzer}

\acrodef{NMSE}{normalized mean-squared error}

\acrodef{ILA}{indirect learning architecture}

\acrodef{mmW} {millimeter wave}
\acrodef{kNN} {k-nearest neighbors}

\acrodef{LS} {least-squares}
\acrodef{rms} {root mean square}
\acrodef{GLA} {generalized Lloyd's algorithm}
\acrodef{ITU} {International Telecommunication Union}
\acrodef{VS-GMP} {vector-switched generalized memory polynomial}
\acrodef{SVS-GMP} {statistical VS-GMP}
\acrodef{I} {in-phase}
\acrodef{Q} {quadrature}

\acrodef{ACLR}{adjacent channel leakage ratio}
\acrodef{ACPR}{adjacent channel power ratio}
\acrodef{AM/AM}{amplitude-to-amplitude}
\acrodef{AM/PM}{amplitude-to-phase}
\acrodef{AI}{artificial intelligence}
\acrodef{ANN}{artificial neural network}
\acrodef{LM}{levenberg–marquardt}

\acrodef{ReLU}{rectified linear unit}
\acrodef{Tanh}{tangent sigmoid}
\acrodef{ELU}{exponential linear unit}
\acrodef{GELU}{gaussian error linear unit}

\acrodef{CDF}{cumulative distribution function}
\acrodef{CCDF}{complementary cumulative distribution function}
\acrodef{CCRR}{computational complexity reduction ratio}

\acrodef{DT} {decision tree}
\acrodef{DT-GMP} {decision tree-based generalized memory polynomial}
\acrodef{DL}{deep learning}
\acrodef{DUT} {device under test}
\acrodef{DLA}{direct learning architecture}
\acrodef{DTT}{digital terrestrial television}
\acrodef{DDR}{dynamic deviation reduction}
\acrodef{EMP}{envelope memory polynomial}

\acrodef{EVM} {error vector magnitude}

\acrodef{FR1}{frequency range 1}
\acrodef{FR2}{frequency range 2}
\acrodef{FLOP}{floating point operation}

\acrodef{LTE}{long term evolution}
\acrodef{LTI}{linear time-invariant}

\acrodef{MLP} {multilayer perceptron}
\acrodef{MSE} {mean-squared error}

\acrodef{NN} {neural network}
\acrodef{RNN} {recurrent NN}
\acrodef{CNN} {convolutional NN}
\acrodef{STD}{standard deviation}

\acrodef{RVTDNN} {real-valued time-delay NN}
\acrodef{R2TDNN} {residual real-valued time-delay NN}
\acrodef{ARVTDNN} {augmented real-valued time-delay NN}
\acrodef{SRTDNN} {simplified real-valued time-delay NN}
\acrodef{RVTDCNN} {real-valued time-delay convolutional NN}
\acrodef{LSTM} {long short-term memory}
\acrodef{TDNN} {time-delay neural network}
\acrodef{$I/Q$} {in-phase and quadrature}
\acrodef{ResNet} {residual neural network}

\acrodef{MMSE} {minimum mean-squared error}
\acrodef{MVUE} {minimum variance unbiased estimate}
\acrodef{LTI} {linear time-invariant}
\acrodef{BIBO} {bounded input and bounded output}

\acrodef{LHS} {left-hand side}
\acrodef{RHS} {right-hand side}

\acrodef{MI} {memory index}
\acrodef{CRLB} {Cramér–Rao lower bound}

\acrodef{NARMA} {nonlinear auto-regressive moving average}

\acrodef{GOD} {gap-induced orthogonality delimitation}
\acrodef{i.i.d.} {independent and identically distributed}
\acrodef{PIML} {physics-informed machine learning}

\acrodef{ROC} {receiver operating characteristic}